\documentclass[11pt]{article}

\usepackage{amssymb}
\usepackage{graphicx}
\usepackage{color}
\usepackage{amsmath}
\usepackage{url}
\usepackage{color}
\usepackage{hyperref}
\usepackage{paralist}
\usepackage{amsmath,amsfonts,amssymb,amsthm}
\usepackage{mathtools}
\usepackage{enumerate}
\usepackage{indentfirst}
\usepackage{amsopn}
\usepackage{fancyhdr}
\usepackage{rotating}
\usepackage{overpic}
 \usepackage{authblk}

\usepackage[colorinlistoftodos]{todonotes}

\newtheorem{theorem}{Theorem}[section]

\newtheorem{lemma}[theorem]{Lemma}
\newtheorem{corollary}[theorem]{Corollary}
\newtheorem{definition}[theorem]{Definition}

\newcommand{\dS}{$d$-\textsc{Scattered Set}}

\newcommand{\cw}{\ensuremath\textrm{cw}}
\newcommand{\tw}{\ensuremath\textrm{tw}}

\title{Improved (In-)Approximability Bounds for $d$-Scattered Set}
\author{Ioannis Katsikarelis}
\author{Michael Lampis}
\author{Vangelis Th.\ Paschos}
\affil{Université Paris-Dauphine, PSL Research University, CNRS, UMR 7243 \\ LAMSADE, 75016, Paris, France \texttt{\{ioannis.katsikarelis|michail.lampis|paschos\}@lamsade.dauphine.fr}}
 \date{}

\begin{document}
\maketitle

\begin{abstract}
In the \dS\ problem we are asked to select at least $k$ vertices of a given graph, so that the distance between any pair is at least $d$. We study the problem's (in-)approximability and offer improvements and extensions of known results for \textsc{Independent Set}, of which the problem is a generalization.
Specifically, we show:
\begin{itemize}
 \item A lower bound of $\Delta^{\lfloor d/2\rfloor-\epsilon}$ on the approximation ratio of any polynomial-time algorithm for graphs of maximum degree $\Delta$ and an improved upper bound of $O(\Delta^{\lfloor d/2\rfloor})$ on the approximation ratio of any greedy scheme for this problem.
 \item A polynomial-time $2\sqrt{n}$-approximation for bipartite graphs and even values of $d$, that matches the known lower bound by considering the only remaining case.
 \item A lower bound on the complexity of any $\rho$-approximation algorithm of (roughly) $2^{\frac{n^{1-\epsilon}}{\rho d}}$ for even $d$ and $2^{\frac{n^{1-\epsilon}}{\rho(d+\rho)}}$ for odd $d$ (under the randomized ETH), complemented by $\rho$-approximation algorithms of running times that (almost) match these bounds.
\end{itemize}
\end{abstract}

\section{Introduction}\label{sec_intro}

In this paper we study the \dS\ problem: given graph $G=(V,E)$, we are asked if there exists a set $K$ of at least $k$ \emph{selections} from $V$, such that the distance between any pair $v,u\in K$ is at least
$d(v,u)\ge d$, where $d(v,u)$ denotes the shortest-path distance from $v$ to $u$.
The problem can already be seen to be hard as it generalizes \textsc{Independent Set} (for $d=2$) and thus the optimal $k$ cannot be approximated to $n^{1-\epsilon}$ in polynomial time \cite{Hastad96} (under standard complexity assumptions), while an alternative name is \textsc{Distance-$d$ Independent Set} \cite{EtoGM14,MontealegreT16}.

The problem has been well-studied, also from the parameterized point of view \cite{katsikarelisWG,Pilipczuk18}, while approximability in polynomial time has already been considered for bipartite, regular and degree-bounded graphs \cite{EtoILM16,EtoGM14}, perhaps the natural candidate for the next intractability frontier. This paper aims to advance our understanding in this direction by providing the first lower bound on the approximation ratio of any polynomial-time algorithm as a function of the maximum degree of any vertex in the input graph, while also improving upon the known ratios to match this lower bound. On bipartite graphs, our aim is to complete the picture by considering the only remaining open case for this class, 
before turning our attention to super-polynomial running times with the purpose of extending known upper/lower bounds for \textsc{Independent Set}.
%

Before moving on to describe our results in detail, we note that these may be dependent on the parity of our distance parameter $d$ as being even or odd. Both our running times and ratios can be affected by this peculiarity of the problem that, intuitively, arises due to the (non)existence of a \emph{middle} vertex on a path of length $d$ between two endpoints: if $d$ is even then such a vertex can exist at equal distance $d/2$ from any number of vertices in the solution, while if $d$ is odd there can be no vertex at equal distance from any pair of vertices in the solution. This idiosyncrasy can change the way in which both our algorithms and hardness constructions work and in some cases even entirely alters the problem's complexity (e.g.\ in the results of \cite{EtoGM14}). 

\paragraph{Our contribution:} Section~\ref{dSS_poly-sec} concerns itself with strictly polynomial running times. We first show that there is no polynomial-time approximation algorithm for \dS\ with ratio $\Delta^{\lfloor d/2\rfloor-\epsilon}$ in graphs of maximum degree $\Delta$. Our complexity assumption is NP$\not\subseteq$BPP due to our use as a starting point of a randomized construction for \textsc{Independent Set} by \cite{Chalermsook13}, that we then build upon to produce highly efficient (in terms of maximum vertex degree and diameter) instances of \dS. This is the first lower bound that considers $\Delta$ and generalizes the known $\Delta^{1-\epsilon}$-inapproximability of \textsc{Independent Set} (see Theorem~5.2 of \cite{Chalermsook13}, restated here as Theorem~\ref{gap_reduction}, as well as \cite{Alon1995}). Maximum vertex degree $\Delta$ plays an important role in the context of independence (e.g.\ \cite{Berman99,DEMANGE1997105, Halldorsson1997}) and was specifically studied for \dS\ in \cite{EtoILM16}, where polynomial-time $O(\Delta^{d-1})$- and $O(\Delta^{d-2}/d)$-approximations are given. We improve upon these upper bounds by showing that any degree-based greedy approximation algorithm in fact achieves a ratio of $O(\Delta^{\lfloor d/2\rfloor})$, also matching our lower bound.
We then turn our attention to bipartite graphs and show that \dS\ can be approximated within a factor of $2\sqrt{n}$ in polynomial time also for even values of $d$, matching its known $n^{1/2-\epsilon}$-inapproximability from \cite{EtoGM14} and complementing the known $\sqrt{n}$-approximation for odd values of $d$ from \cite{HalldorssonKT00}.

Section~\ref{dSS_super-poly-sec} follows this up by considering super-polynomial running times, presenting first an exact exponential-time algorithm for \dS\ of complexity $O^*((ed)^{\frac{2n}{d}})$ based on a straightforward upper bound on the size of any solution and then considering the inapproximability of the problem in the same complexity range. We show that no $\rho$-approximation algorithm can take time (roughly) $2^{\frac{n^{1-\epsilon}}{\rho d}}$ for even $d$ and $2^{\frac{n^{1-\epsilon}}{\rho(d+\rho)}}$ for odd $d$, under the (randomized) ETH. This is complemented by (almost) matching $\rho$-approximation algorithms of running times $O^*((e\rho d)^{\frac{2n}{\rho d}})$ for even $d$ and $O^*((e\rho d)^{\frac{2n}{\rho(d+\rho)}})$ for odd $d$. We note that the current state-of-the-art PCPs are unable to distinguish between optimal running times of the form $2^{n/\rho}$ and $\rho^{n/\rho}$ for $\rho$-approximation algorithms, due to the poly-logarithmic factor added by even the most efficient constructions and we thus do not focus on the poly-logarithmic factors differentiating our upper and lower bounds. These results provide a complete characterization of the optimal relationship between the worst-case approximation ratio $\rho$ achievable for \dS\ by any algorithm, its running time and the distance parameter $d$, \emph{for any point in the trade-off curve}, in a similar manner as was done for \textsc{Independent Set} in \cite{Chalermsook13,cyg08} (see also \cite{BonnetLP16,BourgeoisEP11}), by also considering the range of possible values for $d$. We observe that the distance parameter $d$ acts as a scaling factor for the size of the instance, whereby the problem becomes easier when vertices are required to be much further apart, a feature counterbalanced by the chosen approximation ratio $\rho$, with small values guaranteeing the quality of the produced solutions, yet also negatively impacting on the exponent of the running time.

We close the paper with a supplementary note on the treewidth of power graphs obtained through observations related to our previous results, while Section~\ref{sec_conc} provides some concluding remarks and discussion on open problems. Our results are also summarized in Table~\ref{poly_approx_table} below.

\begin{table}[htpb]\centering
\begin{tabular}{|c||c|c|}
\hline
                 & Inapproximability & Approximation \\ \hline\hline
Super-polynomial &$2^{\frac{n^{1-\epsilon}}{\rho d}}$ (\ref{inapprox_even})/ $2^{\frac{n^{1-\epsilon}}{\rho(d+\rho)}}$(\ref{inapprox_odd})& $O^*((e\rho d)^{\frac{2n}{\rho d}})$ (\ref{approx_even})/ $O^*((e\rho d)^{\frac{2n}{\rho(d+\rho)}})$ (\ref{approx_odd})\\ \hline
Polynomial       & $\Delta^{\lfloor d/2\rfloor-\epsilon}$ (\ref{delta_inapprox})& $O(\Delta^{\lfloor d/2\rfloor})$ (\ref{delta_ratio})\\ \hline
Bipartite graphs & $n^{1/2-\epsilon}$ \cite{EtoGM14} & $2\sqrt{n}$ (\ref{bipartite_even_thm})\\ \hline
\end{tabular}\caption{A summary of our results (theorem numbers), for even/odd values of $d$.}\label{poly_approx_table}
\end{table}

\paragraph{Related work:}
 Eto et al.\ (\cite{EtoILM16}) showed that on $r$-regular graphs the problem is APX-hard for $r,d\ge3$, while also
providing polynomial-time $O(r^{d-1})$-approximations. They also show a
polynomial-time 2-approximation on cubic graphs and a polynomial-time
approximation scheme (PTAS) for planar graphs and every fixed constant $d\ge3$. For a class of graphs with at
most a polynomial (in $n$) number of minimal separators, \dS\ can be solved in
polynomial time for even $d$, while it remains NP-hard on chordal graphs and
any odd $d\ge3$ \cite{MontealegreT16}. For bipartite graphs, the problem is
NP-hard to approximate within a factor of $n^{1/2-\epsilon}$ and W[1]-hard for
any fixed $d\ge3$.
Further, for any odd $d\ge3$, it remains NP-complete, inapproximable and
W[1]-hard \cite{EtoGM14}. It is NP-hard even for planar bipartite graphs
of maximum degree 3, yet a 1.875-approximation is available on cubic graphs
\cite{EtoILM17}. Furthermore, \cite{FominLRS11} shows the problem admits an
EPTAS on (apex)-minor-free graphs, based on the theory of bidimensionality,
while on a related result \cite{MarxP15} offers an
$n^{O(\sqrt{n})}$-time algorithm for planar graphs, making use of Voronoi diagrams and based on ideas previously used to obtain geometric QPTASs.
Finally, \cite{katsikarelisWG} presents tight upper/lower bounds on the structurally parameterized complexity of the problem, while \cite{Pilipczuk18} shows that it admits an almost linear kernel on every nowhere dense graph class.

\section{Definitions and Preliminaries}
We use standard graph-theoretic notation. For a graph $G=(V,E)$, we let $V(G)\coloneqq V$ and $E(G)\coloneqq E$, an edge $e\in E$ between $u,v\in V$ is denoted by $(u,v)$ and for a subset $X\subseteq V$,
$G[X]$ denotes the graph induced by $X$.
We let $d_G(v,u)$ denote the shortest-path distance (i.e.\ the number of edges) from $v$ to $u$ in $G$.
We may omit subscript $G$ if it is clear from the context.
The maximum distance between vertices is the \emph{diameter} of the graph, while the minimum among all the maximum distances between a vertex to all other vertices (their \emph{eccentricities}) is considered as the \emph{radius} of the graph.

For a vertex $v$, we let $N^d_G(v)$ denote the (open) \emph{$d$-neighborhood} of $v$ in $G$, i.e.\ the set of vertices at distance $\le d$ from $v$ in $G$ (without $v$), while for a subset $U\subseteq V$, $N^d_G(U)$ denotes the union of the $d$-neighborhoods of vertices $u\in U$. In a graph $G$ whose maximum degree is bounded by $\Delta$, the size of the $d$-neighborhood of any vertex $v$ is upper bounded by the well-known Moore bound: $|N_G^d(v)|\le\Delta\sum_{i=0}^{d}(\Delta-1)^i$. For an integer $q$, the \emph{$q$-th power graph of $G$}, denoted by $G^q$, is defined as the graph obtained from $G$ by adding to $E(G)$ all edges between vertices $v,u\in V(G)$ for which $d_G(v,u)\le q$.
Furthermore, we let $OPT_d(G)$ denote the maximum size of a $d$-scattered set in $G$ and $\alpha(G)=OPT_2(G)$ denote the size of the largest independent set.

We use $\log(n)$ to denote the base-2
logarithm of $n$, while $\log_{1+\delta}(n)$ is the logarithm
base-$(1+\delta)$, for $\delta>0$.  Recall also that
$\log_{1+\delta}(n)=\log(n)/\log(1+\delta)$.
The functions $\lfloor x\rfloor$ and $\lceil x\rceil$, for $x\in\mathbb{R}$, denote the maximum integer that is not larger and the minimum integer that is not smaller than $x$, respectively.


The Exponential Time Hypothesis (ETH) \cite{ImpagliazzoP01,ImpagliazzoPZ01}
implies that 3-\textsc{SAT} cannot be solved in time $2^{o(n)}$ on instances
with $n$ variables (a slightly weaker statement), while the definition can also refer to \emph{randomized} algorithms.
Finally, we recall here the following result by \cite{Chalermsook13} that some of our reductions will be relying on (slightly paraphrased, see also \cite{BonnetLP16}), that can furthermore be seen as implying the $\Delta^{1-\epsilon}$-inapproximability of \textsc{Independent Set} in polynomial time:

\begin{theorem}[\cite{Chalermsook13}, Theorem~5.2]\label{gap_reduction}
 For any sufficiently small $\epsilon>0$ and any $r\le N^{5+O(\epsilon)}$, there is a randomized polynomial reduction that builds from a formula $\phi$ of \textsc{SAT} on $N$ variables a graph $G$ of size $n=N^{1+\epsilon}r^{1+\epsilon}$ and maximum degree $r$, such that with high probability:
 \begin{itemize}
  \item If $\phi$ is satisfiable, then $\alpha(G)\ge N^{1+\epsilon}r$.
  \item If $\phi$ is not satisfiable, then $\alpha(G)\le N^{1+\epsilon}r^{2\epsilon}$.
 \end{itemize}
\end{theorem}

\section{Polynomial Time}\label{dSS_poly-sec}
We begin by focusing on the behaviour of the problem in the context of strictly polynomial-time approximation. We first examine graphs of bounded degree and provide a tight bound on the achievable approximation ratio, before turning to bipartite graphs in order to finalize the classification in terms of approximability by considering the only open remaining case (when $d$ is even).

\subsection{Inapproximability}\label{dSS_poly-inapprox-sec}
We show that for sufficiently large $\Delta$ and any $\epsilon_1>0, d\ge4$, the \dS\ problem is inapproximable to $\Delta^{\lfloor d/2\rfloor-\epsilon_1}$ on graphs of degree bounded by $\Delta$, unless NP$\subseteq$BPP. Let us first summarize our reduction. Starting from an instance of \textsc{Independent Set} of bounded degree, we create an instance of \dS\ where the degree is (roughly) the $d/2$-th square root of that of the original instance. As we are able to maintain a direct correspondence of solutions in both instances, the $\Delta^{1-\epsilon'}$-inapproximability of IS implies the $\Delta^{\lfloor d/2\rfloor-\epsilon_1}$-inapproximability of \dS.

The technical part of our reduction involves preserving the adjacency between vertices of the original graph without increasing the maximum degree (too far) beyond $\Delta^{2/d}$. We are able to construct a regular tree as a gadget for each vertex and let the edges of the leaves (their total number being equal to $\Delta$) represent the edges of the original graph. To ensure that our gadget has some useful properties (i.e.\ small diameter), we overlay a number of extra edges on each level of the tree (i.e.\ between vertices at equal distance from the root), only sacrificing a small increase in maximum degree. Our complexity assumption is NP$\not\subseteq$BPP, since for the $\Delta^{1-\epsilon'}$-inapproximability of IS we use the randomized reduction from SAT of \cite{Chalermsook13} (Theorem~\ref{gap_reduction} above). In particular, we will prove the following theorem:

\begin{theorem}\label{delta_inapprox}
 For sufficiently large $\Delta$ and any $d\ge4, \epsilon\in(0,\lfloor d/2\rfloor)$, there is no polynomial-time approximation algorithm for \dS\ with ratio $\Delta^{\lfloor d/2\rfloor-\epsilon}$ for graphs of maximum degree $\Delta$, unless NP$\subseteq$BPP.
\end{theorem}

\paragraph{Construction:}
Let $\delta=\left\lceil\sqrt[\lfloor d/2\rfloor]{\Delta}\right\rceil$. Given $\epsilon_1\in(0,\lfloor d/2\rfloor)$ and an instance of \textsc{Independent Set} $G=(V,E)$, where the degree of any vertex is bounded by $\Delta$, we will construct an instance $G'=(V',E')$ of \dS, where the degree is bounded\footnote{We note that this value of $\epsilon_2$ is for odd values of $d$. For $d$ even, the correct value is such that we have the (slightly lower) bound $\delta^{1+\epsilon_2}=\delta+3\delta^{1+2\epsilon_1/d}$, but we write $\epsilon_2$ for both cases to simplify notation.} by $\delta^{1+\epsilon_2}=6\delta^{1+2\epsilon_1/d}$, for $\epsilon_2=2\epsilon_1/d+\log_{\delta}3>2\epsilon_1/d$, while $OPT_2(G)=OPT_d(G')$. We assume $\Delta$ is sufficiently large for $\epsilon_1\ge\frac{d(\log(\log(\Delta))+c)}{4\log(\Delta)/d}$, where $c\le10$ is a small constant, for reasons that become apparent in the following.

Our construction for $G'$ builds a gadget $T(v)$ for each vertex $v\in V$. For even $d$, each gadget $T(v)$ is composed of a $(\delta+1)$-regular tree of height $d/2-1$ and we refer to vertices of $T(v)$ at distance exactly $i$ from the root $t_v$ as being in the $i$-th \emph{height-level} of $T(v)$, letting each such subset be denoted by $T_i(v)$. That is, every vertex of $T_i(v)$ has one neighbor in $T_{i-1}(v)$ (its parent) and $\delta$ neighbors in $T_{i+1}(v)$ (its children). For odd values of $d$, the difference is in the height of each tree being $\lfloor d/2\rfloor$ instead of $d/2-1$. 

Since for even $d$ the number of leaves of $T(v)$ is $\delta^{d/2-1}=(\Delta^{2/d})^{d/2-1}=\Delta^{1-2/d}$ and each such leaf also has $\delta=\Delta^{2/d}$ edges, the number of edges leading outside each gadget is $\delta^{d/2}=\Delta$ and we let each of them correspond to one edge of the original vertex $v$ in $G$, i.e.\ we add an edge between a leaf $x_v$ of $T(v)$ and a leaf $y_u$ of $T(u)$, if $(v,u)\in E$. For odd $d$, the number of leaves is $\left(\left\lceil\sqrt[\lfloor d/2\rfloor]{\Delta}\right\rceil\right)^{\lfloor d/2\rfloor}$ (i.e.\ at least $\Delta$) and we let each leaf correspond to an edge of the original vertex $v$ in $G$, i.e.\ we \emph{identify} two such leaves $x_v,y_u$ of two gadgets $T(v),T(u)$, if $(v,u)\in E$ in $G$. In this way, the gadgets $T(v),T(u)$ share a common ``leaf'' of degree 2, that is at distance $\lfloor d/2\rfloor$ from both roots $t_v\in T(v),t_u\in T(u)$. See Figure \ref{fig:degree_tree} for an illustration.

\begin{figure}[htbp]
 \centerline{\includegraphics[width=140mm]{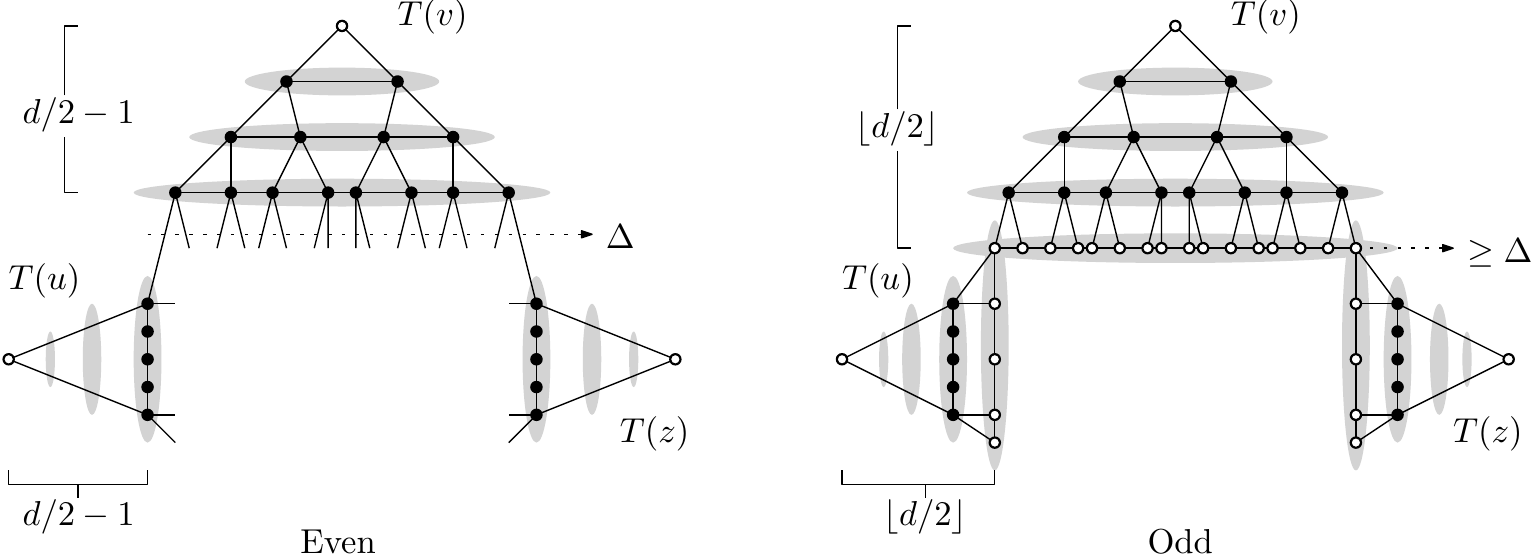}}
 \caption{Our constructions for an example subgraph consisting of a path on three vertices $(u,v,z)$ and even/odd $d$. Ellipses in grey designate the overlaid edges on each height-level.}
 \label{fig:degree_tree}
 \end{figure}

Next, in order to make the diameter of our gadgets at most equal to their height, we will add a number of edges between the vertices of each height-level $i$ of each gadget $T(v)$, for every $v\in V$. We will first add the edges of a cycle plus a random matching (using a technique from \cite{Bollobasv88}) and then the edges of an appropriately chosen power graph of this subgraph containing the edges of the cycle plus the matching.
These edges will be overlaid on each height-level, meaning our final construction will contain the edges of the tree, the cycle, the matching, as well as the power graph.

For even $d$ and each gadget $T(v)$, we first make all vertices $T_i(v)$ at each height-level $i<1+\epsilon_2$ into a clique. For larger height-levels $i\in[1+\epsilon_2,d/2-1]$, we first make the vertices into a cycle (arbitrarily ordered) and then also add a random matching, i.e.\ the edges between each pair of a random partition of $T_i(v)$ into disjoint pairs (plus a singleton if $|T_i(v)|$ is odd). Letting $P_i(v)$ denote these edges of the cycle plus the matching for each $T_i(v)$, we define the subgraph $H_i(v)=(T_i(v),P_i(v))$ and compute the $\lceil((1+2\epsilon_1/d)\log_2(\delta))\rceil$-power graph of $H_i(v)$, finally also adding its edges to $G'$. For odd $d$ and each gadget $T(v)$, we again make the vertices $T_i(v)$ at each height-level $i<1+\epsilon_2$ into a clique and for larger height-levels $i\in[1+\epsilon_2,\lfloor d/2\rfloor]$, we follow the same process.

This concludes our construction, while to prove our claims on the diameter of our gadgets we also make use of the following statements:

\begin{theorem}[\cite{Bollobasv88}, Theorem~1]\label{exp_match_thm}
Let $G$ be a graph formed by adding a random matching to an $n$-cycle. Then with probability tending to 1 as $n$ goes to infinity, $G$ has diameter upper-bounded by $ \log_2(n)+\log_2(\log(n))+c$, where $c$ is a small constant (at most 10).
\end{theorem}

\begin{lemma}\label{exp_power_lem}
 Let $G$ be a graph of diameter $\le a$. Then the diameter of the $b$-power graph $G^b$ is $\le \lceil a/b\rceil$, for any integer $b<a$.
\end{lemma}
\begin{proof}
Consider a path $P$ of (maximum) length $\le a$ between two vertices $v,u$ in $G$. Taking the $b$-power of $G$ adds all edges between vertices of $P$ at distance $\le b$. This means vertex $v$ will be adjacent in $G^b$ to a vertex $x_1$ on $P$ that was at distance $b$ from $u$ in $G$. This vertex $x_1$ will be in turn adjacent to another vertex $x_2$ on $P$ that was at distance $2b$ from $v$ in $G$. Carrying on like this we can find a sequence $x_1,\dots,x_i$ of vertices of $P$, each at distance $b$ in $G$ from its predecessor and follower in the sequence, that form a path $P'$ in $G^b$. Since the length of $P$ is at most $a$, the maximum number $i$ of vertices in the sequence until we reach $u$ is $\lceil a/b\rceil$, giving the length of $P'$ in $G^b$.
\end{proof}

We are now ready to argue about the maximum degree of any vertex in the instances built by our construction.

\begin{lemma}\label{gadget_degree}
 The maximum degree of any vertex in $G'$ is $\le\delta+3\delta^{1+2\epsilon_1/d}$ for even $d$ and $\le6\delta^{1+2\epsilon_1/d}$ for odd $d$.
\end{lemma}
\begin{proof}
 Observe that for $d$ even, the degree of any vertex is bounded by the sum of the $\delta+1$ edges of the tree plus the number of edges added by the power graph (including the three edges of the cycle and matching): $\sum_{k=0}^{\lceil(1+2\epsilon_1/d)\log_2(\delta)\rceil}(3\cdot2^k)=3\cdot2^{\lceil(1+2\epsilon_1/d)\log_2(\delta)-1\rceil}-1\le3\cdot2^{(1+2\epsilon_1/d)\log_2(\delta)}-1=3\cdot\delta^{1+2\epsilon_1/d}-1$, for a total of $\delta+3\delta^{1+2\epsilon_1/d}$.

For $d$ odd, we note that the degree of all other vertices is strictly lower than that of the ``shared'' leaves between gadgets, since each leaf between two gadgets $T(v),T(u)$ (representing the edge $(v,u)$ of $G$) will belong to two subgraphs $H_{\lfloor d/2\rfloor}(v)$ and $H_{\lfloor d/2\rfloor}(u)$. Thus their degree will be $2+2(3\cdot\delta^{1+2\epsilon_1/d}-1)=6\delta^{1+2\epsilon_1/d}$.
\end{proof}

We then bound the diameter of our gadgets in order to guarantee that the solutions in our reduction will be well-formed. Our statement is probabilistic and conditional on our assumption on the size of $\Delta$ as being sufficiently large.

\begin{lemma}\label{gadget_diameter}
 With high probability, the diameter of each gadget $T(v)$ is $d/2-1$ for even $d$ and $\lfloor d/2\rfloor$ for odd $d$, for sufficiently large $\Delta$.
\end{lemma}
\begin{proof}
 First, observe that for sufficiently large $n$, $c\le10$ and $\epsilon_1\in(0,\lfloor d/2\rfloor)$, it is $\log_2(\log(n))+c<(2\epsilon_1/d)\log_2(n)$. For even $d$, our construction uses $n$-cycles of length $n=\delta^i$ for each $i\in[1+\epsilon_2,d/2-1]$, meaning that $\Delta$ must be sufficiently large for $\epsilon_1\ge\frac{d(\log(2i\log(\Delta)/d)+c)}{4i\log(\Delta)/d}$, while for odd $d$ it is $i\in[1+\epsilon_2,\lfloor d/2\rfloor]$. As noted above, our assumption for $\Delta$ requires that it is sufficiently large for $\epsilon_1\ge\frac{d(\log(\log(\Delta))+c)}{4\log(\Delta)/d}$, which is $>\frac{d(\log(2i\log(\Delta)/d)+c)}{4i\log(\Delta)/d}$ for the required range of $i$ in both cases.
 
 By Theorem~\ref{exp_match_thm}, the distance between any pair of vertices at height-level $i$ after adding the edges of $P_i(v)$ is at most $\log_2(\delta^i)+\log_2(\log(\delta^i))+c$ (with high probability). This is $<(1+2\epsilon_1/d)\log_2(\delta^i)$ for sufficiently large $\Delta$. By Lemma~\ref{exp_power_lem}, taking the $\lceil((1+2\epsilon_1/d)\log_2(\delta))\rceil$-power of $H_i(v)$ shortens the distance to at most $\frac{(1+2\epsilon_1/d)\log_2(\delta^i)}{\lceil(1+2\epsilon_1/d)\log_2(\delta)\rceil}\le i$, for each height-level $i\in[1+\epsilon_2,d/2-1]$. For smaller values of $i$, the vertices of each height-level form a clique and the distance between any pair of them is thus at most 1.
 
 For odd values of $d$, the size $n$ of the cycles we use is again $\delta^i$, with $i\in[1+\epsilon_2,\lfloor d/2\rfloor]$ and we thus have once more that for sufficiently large $\Delta$ the distance between any pair of vertices after adding the edges of $P_i(v)$ to each height-level $i$ of each gadget $T(v)$ is at most $(1+2\epsilon_1/d)\log_2(\delta^i)$ (with high probability) and at most $i$ after taking the $\lceil((1+2\epsilon_1/d)\log_2(\delta))\rceil$-power of $H_i(v)$. Again, for smaller $i<1+\epsilon_2$, $T_i(v)$ is a clique.
 
 Since at each height-level $i$, no pair of vertices is at distance $>i$ with $i\le d/2-1$ for even $d$ and $i\le\lfloor d/2\rfloor$ for odd $d$, the distance between any vertex $x$ at some height-level $i_x$ to another vertex $y$ at height-level $i_y>i_x$ will be at most $i_x$ from $x$ to the root of the subtree of $T(v)$ (at level $i_x$) that contains $y$. From there to $y$ it will be at most $d/2-1-i_x$ for even $d$ and at most $\lfloor d/2\rfloor-i_x$ for odd $d$. Furthermore, the distance from the root of $T(v)$ to a leaf is exactly $d/2-1$ for even $d$ and exactly $\lfloor d/2\rfloor$ for odd $d$.
\end{proof}

We finalize our argument with a series of lemmas leading to the proof of Theorem~\ref{delta_inapprox}, that detail the behaviour of solutions that can form in our construction, relative to the independence of vertices in the original graph.

\begin{lemma}\label{gadget_exclusion}
 No \dS\ in $G'$ can contain a vertex from gadget $T(v)$ and a vertex from gadget $T(u)$, if $(u,v)\in E$.
\end{lemma}
\begin{proof}
 Since $(u,v)\in E$, there is an edge $(x_v,y_u)\in E'$ between a leaf $x_v\in T(v)$ and $y_u\in T(u)$ for even $d$, while for for odd $d$ the leaf $x$ belongs to both $T(v),T(u)$ and is at distance $\lfloor d/2\rfloor$ from each of their roots. Thus for even $d$ the maximum distance from any vertex of $T(v)$ to $y_u\in T(u)$ is $d/2-1+1=d/2$, by Lemma~\ref{gadget_diameter}, and for odd $d$ this is $\lfloor d/2\rfloor$. Since, by the same lemma, the diameter of $T(u)$ is $d/2-1$ for even $d$ and $\lfloor d/2\rfloor$ for odd $d$, there is no vertex of $T(u)$ that can be in any \dS\ along with any vertex of $T(v)$, as the maximum distance is $\le d/2+d/2-1=d-1$ for even $d$ and $\le\lfloor d/2\rfloor+\lfloor d/2\rfloor=d-1$ for odd $d$.
\end{proof}

\begin{lemma}\label{gadget_inclusion}
 If $(u,v)\notin E$, then the distance between the root $t_v$ of $T(v)$ and the root $t_u$ of $T(u)$ is at least $d$.
\end{lemma}
\begin{proof}
 Since $(u,v)\notin E$, then there is no edge between any pair of leaves $x_u$ of $T(u)$ and $y_v$ of $T(v)$ for even $d$. Thus the shortest possible distance between any such pair of leaves is 2 for even $d$, through a third leaf $z_w$ of another gadget $T(w)$ corresponding to a vertex $w$ adjacent to both $u$ and $v$ in $G$. The distance from $t_v$ to any leaf of $T(v)$ is $d/2-1$ and the distance from $t_u$ to any leaf of $T(u)$ is also $d/2-1$. Thus the distance from $t_u$ to $t_v$ must be at least $d/2-1+d/2-1+2=d$.
 
 For odd $d$, there is no shared leaf $x$ between the two gadgets, i.e.\ at distance $\lfloor d/2\rfloor$ from both roots.
 Thus the distance between two leaves $x_u\in T(u)$ and $y_v\in T(v)$ is at least 1, if each of these is shared with a third gadget $T(w)$ corresponding to a vertex $w$ that is adjacent to both $u$ and $v$ in $G$. The distance from $t_v$ to any leaf of $T(v)$ is $\lfloor d/2\rfloor$ and the distance from $t_u$ to any leaf of $T(u)$ is also $\lfloor d/2\rfloor$. Thus the distance from $t_u$ to $t_v$ is at least $2\lfloor d/2\rfloor+1=d$. 
\end{proof}

\begin{lemma}\label{delta_FWD}
 For any independent set $S$ in $G$, there is a \dS\ $K$ in $G'$, with $|S|=|K|$.
\end{lemma}
\begin{proof}
 Given an independent set $S$ in $G$, we let $K$ include the root vertex $t_v\in T(v)$ for each $v\in S$. Clearly, $|S|=|K|$. Since $S$ is independent, there is no edge $(u,v)$ between any pair $u,v\in S$ and thus, by Lemma~\ref{gadget_inclusion}, vertices $t_v$ and $t_u$ are at distance at least $d$.
\end{proof}

\begin{lemma}\label{delta_BWD}
 For any \dS\ $K$ in $G'$, there is an independent set $S$ in $G$, with $|K|=|S|$.
\end{lemma}
\begin{proof}
 Given a \dS\ $K$ in $G$, we know there is at most one vertex from each gadget $T(v)$ in $K$, since its diameter is $d/2-1$ for even $d$ and $\lfloor d/2\rfloor$ for odd $d$, by Lemma~\ref{gadget_diameter}. Furthermore, for any two vertices $x,y\in K$, we know by Lemma~\ref{gadget_exclusion} that if $x\in T(u)$ and $y\in T(v)$ for gadgets corresponding to vertices $u,v\in V$, then $(u,v)\notin E$ and thus $u,v$ are independent in $G$. We let set $S$ contain each vertex $v\in V$ whose corresponding gadget $T(v)$ contains a vertex of $K$. These vertices are all independent and also $|K|=|S|$.
\end{proof}

\begin{proof}[Proof of Theorem~\ref{delta_inapprox}]
 We suppose the existence of a polynomial-time approximation algorithm for \dS\ with ratio $\Delta^{\lfloor d/2\rfloor-\epsilon_1}$ for graphs of maximum degree $\Delta$ and some $0<\epsilon_1<d/2$. We assume $\Delta$ is sufficiently large for $\epsilon_1\ge\frac{d(\log(\log(\Delta))+10)}{4\log(\Delta)/d}$.
 
 Starting from a formula $\phi$ of SAT on $N$ variables, where $N$ is also sufficiently large, i.e\ $N>\Delta^{1/(5+O(\epsilon'))}$ (where $\epsilon'$ is defined below), we use Theorem~\ref{gap_reduction} to produce an instance $G=(V,E)$ of \textsc{Independent Set} on $|V|=N^{1+\epsilon'}\Delta^{1+\epsilon'}$ vertices and of maximum degree $\Delta$, such that with high probability: if $\phi$ is satisfiable, then $\alpha(G)\ge N^{1+\epsilon'}\Delta$; if $\phi$ is not satisfiable, then $\alpha(G)\le N^{1+\epsilon'}\Delta^{2\epsilon'}$. Thus approximating \textsc{Independent Set} in polynomial time on $G$ within a factor of $\Delta^{1-2\epsilon'}$, for $\epsilon'>0$, would permit us to decide if $\phi$ is satisfiable, with high probability.
 
 We next use the above construction to create an instance $G'$ of \dS\ where the degree is bounded by $6\delta^{1+2\epsilon_1/d}=\delta^{1+\epsilon_2}$, for $\epsilon_2>2\epsilon_1/d$, by Lemma~\ref{gadget_degree}. Slightly overloading notation, we let $\epsilon_3\ge\epsilon_2$ be such that $\delta^{1+\epsilon_2}=(\lceil\Delta^{\frac{1}{\lfloor d/2\rfloor}}\rceil)^{1+\epsilon_2}=(\Delta^{\frac{1}{\lfloor d/2\rfloor}})^{1+\epsilon_3}$.
 We now let $\epsilon'=\frac{\epsilon_1(1+\epsilon_3)-\epsilon_3\lfloor d/2\rfloor}{2\lfloor d/2\rfloor}$. Note that $\epsilon'>0$, since $\epsilon_3\ge\epsilon_2>2\epsilon_1/d$.
 
 We then apply the supposed approximation for \dS\ on $G'$. 
 This returns a solution at most $(\delta^{1+\epsilon_2})^{\lfloor d/2\rfloor-\epsilon_1}=\Delta^{(1-\frac{\epsilon_1}{\lfloor d/2\rfloor})(1+\epsilon_3)}=\Delta^{1-\frac{\epsilon_1(1+\epsilon_3)-\epsilon_3\lfloor d/2\rfloor}{\lfloor d/2\rfloor}}=\Delta^{1-2\epsilon'}$ from the optimum. By Lemma~\ref{delta_BWD} we can find a solution for \textsc{Independent Set} in $G$ of the same size, i.e.\ we can approximate $\alpha(G)$ within a factor of $\Delta^{1-2\epsilon'}$, again, with high probability (as Lemma~\ref{gadget_diameter} is also randomized). This would allow us to decide if $\phi$ is satisfiable and thus solve SAT in polynomial time with two-sided bounded errors, implying NP$\subseteq$BPP.
\end{proof}

\subsection{Approximation}\label{dSS_poly-approx-sec}

We next show that any (degree-based) greedy polynomial-time approximation algorithm for \dS\ achieves a ratio of $O(\Delta^{\lfloor d/2\rfloor})$, thus improving upon the analysis of \cite{EtoILM16} and the $O(\Delta^{d-1})$- and $O(\Delta^{d-2}/d)$-approximations given therein.

Our strategy is to bound the size of the largest $d$-scattered set in any graph of maximum degree at most $\Delta$ and radius at most $d-1$, centered on some vertex $v$. The idea is that in one of its iterations our greedy algorithm would select $v$ and thus exclude all other vertices within distance $d-1$ from $v$, yet an upper bound on the size of the largest possible $d$-scattered set can guarantee that the ratio of our algorithm will not be too large. 

The following definition of our ``merge'' operation (see also Figure~\ref{fig:merge}) will allow us to consider all possible graphs of a given radius and degree and provide upper bounds on the size of the optimal solution in such graphs. These bounds on the size of the optimal are then used to compare it to those solutions produced by our greedy scheme.

\begin{figure}[htbp]
 \centerline{\includegraphics[width=90mm]{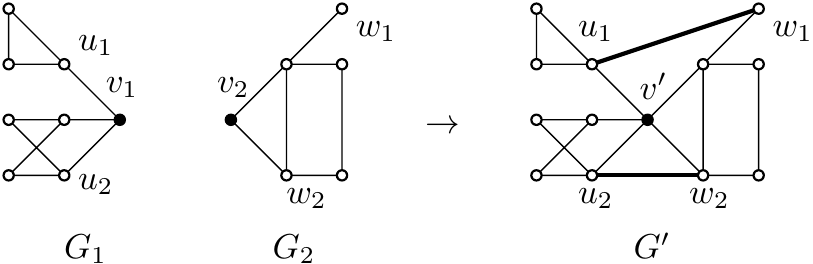}}
 \caption{An example graph $G'=M_{G_1}^{G_2}(v_1,v_2,[u_1,u_2],[w_1,w_2])$ for $G_1,G_2$ shown on the left, with edges added by the third merge operation shown in bold.}
 \label{fig:merge}
 \end{figure}

\begin{definition}[Merge operation]\label{merge_operation}
 For two connected graphs $G_1=(V_1,E_1),G_2=(V_2,E_2)$, the \emph{merged} graph $M_{G_1}^{G_2}(v_1,v_2,{\bf U},{\bf W})$, where ${\bf U}=[u_1,\dots,u_{k_1}]$, ${\bf W}=[w_1,\dots,w_{k_2}]$ are ordered (possibly empty and with repetitions allowed) sequences of vertices from $V_1$ and $V_2$, respectively, is defined as the graph $G'=(V',E')$ obtained by:
 \begin{enumerate}
  \item Identification of vertex $v_1\in G_1$ and vertex $v_2\in G_2$, i.e.\ $V'$ is composed of the union of $V_1,V_2$ after removal of vertices $v_1,v_2$ and addition of a new vertex $v'$.
  \item Replacement of all edges of $v_1,v_2$ by new edges with $v'$ as the new endpoint, i.e.\ $E'$ contains all edges of $E_1,E_2$, where any edges incident on $v_1$ or $v_2$ are now incident on $v'$.
  \item Addition of a number of edges between vertices of $G_1,G_2$, i.e.\ $E'$ also contains one edge between every pair $(u_i,w_j)$ from ${\bf U,W}$, for $i=j$.
 \end{enumerate}
\end{definition}

\begin{lemma}\label{delta_ball}
 The maximum size of a $d$-scattered set in any graph of maximum degree at most $\Delta$ and radius at most $\lceil d/2\rceil$ centered on some vertex $v$ is at most $\Delta$.
\end{lemma}
\begin{proof}
 Consider a graph of maximum degree $\Delta$ and radius $\lceil d/2\rceil$ centered on some vertex $v$: the only pairs of vertices at distance $d$ from each other must be at distance $\ge\lceil d/2\rceil$ from $v$ (or $\lfloor d/2\rfloor$ in one side for odd $d$), as any vertex $u$ at distance $<\lfloor d/2\rfloor$ from $v$ will be at distance $<d$ from any other vertex  $z$ in the graph, since $z$ is at distance $\le\lceil d/2\rceil$ from $v$ (due to the graph's radius).
 Furthermore, for every vertex in the $d$-scattered set, there must be an edge-disjoint path of length at least $\lceil d/2\rceil$ to $v$ that is not shared with any other such vertex, i.e.\ these paths can only share vertex $v$ at distance $\lceil d/2\rceil$ from their endpoints (the vertices that can be in a $d$-scattered set). As the degree of $v$ is bounded by $\Delta$, the number of such disjoint paths also cannot be more than $\Delta$.
\end{proof}

\begin{lemma}\label{ball_join}
 Given two graphs $G_1,G_2$, for $G'= M_{G_1}^{G_2}(v_1,v_2,{\bf U},{\bf W})$ and any ${\bf U},{\bf W}$, it is $OPT_d(G')\le OPT_d(G_1)+OPT_d(G_2)$.
\end{lemma}
\begin{proof}
 Assume for the sake of contradiction that $OPT_d(G')>OPT_d(G_1)+OPT_d(G_2)$ and let $S$ denote an optimum $d$-scattered set in $G'$ of this size, with $S_1=S\cap \{V_1\setminus \{v_1\}\cup \{v'\}\}$ and $S_2=S\cap \{V_2\setminus\{v_2\}\cup\{v'\}\}$ denoting the parts in $G_1,G_2$, respectively. Since $|S|>OPT_d(G_1)+OPT_d(G_2)$, then for any pair of optimal $K_1\subseteq V_1$, $K_2\subseteq V_2$ in $G_1,G_2$, there must be at least one vertex $s$ in $S$ for which $s\notin K_1$ and $s\notin K_2$, but it must be $s\in S_1$ or $s\in S_2$ (or both, if $s=v')$.
 
 Observe that the distance between any pair of vertices in $G_1$ (the same holds for $G_2$) cannot increase in $G'$ after the merge operation, since identification of a pair of vertices between two graphs and addition of any number of edges between the two can only decrease their distance Thus if $s\in V_1$, then $S_1$ is also a $d$-scattered set in $G_1$ (potentially substituting $v'$ for $v_1$) and so, if $|S_1|>|K_1|$ then $K_1$ was not optimal for $G_1$. If $|S_1|\le|K_1|$, it must be $|S_2|>|K_2|$ contradicting the optimality of $K_2$ for $G_2$. Similarly, if $u\in V_2$ we have either $|S_1|>|K_1|$ or $|S_2|>|K_2|$.
 
 If $u$ is the merged vertex $v'$ then there must be at least two other vertices added from $V_1,V_2$ for $|S|>|K_1|+|K_2|$, since $S$ can only contain $v'$ in the place of $v_1\in K_1$ and $v_2\in K_2$. In this case the same argument as above gives the contradiction.
\end{proof}

\begin{lemma}\label{d_ball}
 For any graph $G=(V,E)$ of maximum degree at most $\Delta$ and radius at most $d-1$ centered on some vertex $v$, it is $OPT_d(G)\le O(\Delta^{\lfloor d/2\rfloor})$.
\end{lemma}
\begin{proof}
 Any graph $G$ of maximum degree at most $\Delta$ and radius at most $d-1$ centered on a vertex $v$ can be obtained by the following process: we begin with a graph $H$ of radius at most $\lfloor d/2\rfloor-1$ and maximum degree $\Delta$. Let $\{v_1,\dots,v_k\}\in H$ be the set of vertices at maximum distance from $v$, i.e.\ $d_H(v,v_i)=\lfloor d/2\rfloor-1$. Since the degree of $H$ is bounded by $\Delta$, it must be $k\le \Delta^{\lfloor d/2\rfloor-1}$. We now let $H_i$, for each $i\le k$, denote a series of at most $k$ graphs of radius at most $\lceil d/2\rceil$ centered on a vertex $v_i$ and maximum degree $\Delta$.
 
 Repeatedly applying the merge operation $M_H^{H_i}(v_1,v_i,{\bf U},{\bf W})$ between graph $H$ (or the result of the previous operation) and such a graph $H_i$ we can obtain any graph $G$ of radius at most $d-1$: identifying a vertex $v_j\in H$ (for $j\in[1,k]$) at maximum distance from $v$ with the central vertex $v_i$ of $H_i$ and then adding any number of edges between the vertices of $H$ and $H_i$ (while respecting the maximum degree of $\Delta$), we can produce any graph of radius $\le d-1$, since the distance from $v$ to each $v_j$ is at most $\lfloor d/2\rfloor-1$ and from there to any vertex of $H_i$ it is at most $\lceil d/2\rceil$. The remaining structure of $G$ can be constructed by the chosen structures of the graphs $H,H_i$ and the added edges between them, i.e.\ the sequences ${\bf U,W}$.
  
 By Lemma~\ref{delta_ball} it is $OPT_d(H_i)\le\Delta$ and by Lemma~\ref{ball_join}, it must be $OPT_d(G)\le OPT_d(H)+\sum_{i=1}^{k}OPT_d(H_i)\le1+\Delta\cdot\Delta^{\lfloor d/2\rfloor-1}\le1+\Delta^{\lfloor d/2\rfloor}$.
\end{proof}

\begin{theorem}\label{delta_ratio}
 Any degree-based greedy approximation algorithm for \dS\ achieves a ratio of $O(\Delta^{\lfloor d/2\rfloor})$ on graphs of degree bounded by $\Delta$.
\end{theorem}
\begin{proof}
 Let $G=(V,E)$ be the input graph and consider the process of our supposed greedy algorithm: it picks a vertex $v_i$, removes it from consideration along with the set $V_i\subseteq V$ of vertices at distance at most $d-1$ from $v_i$ and continues the process until there are no vertices left to consider. The sets $V_1,\dots,V_{ALG}$ thus form a partition of $G$. By Lemma~\ref{d_ball}, the optimum size of a $d$-scattered set in any such $V_i$ is at most $O(\Delta^{\lfloor d/2\rfloor})$ and thus $OPT_d(G)\le ALG\cdot O(\Delta^{\lfloor d/2\rfloor})$, by Lemma~\ref{ball_join}, since $G$ can be seen as the merged graph of $G[V_1],\dots,G[V_{ALG}]$.
\end{proof}

\subsection{Bipartite graphs}\label{dSS_bipartite-sec}
Here we consider bipartite graphs and show that \dS\ is approximable to $2\sqrt{n}$ in polynomial time also for even values of $d$. Our algorithm will be applied on both sides of the bipartition and each time it will only consider vertices from one side as candidates for inclusion in the solution. Appropriate sub-instances of \textsc{Set Packing} are then defined and solved using the known $\sqrt{n}$-approximation for that problem.

\begin{definition}
 For a bipartite graph $G=(A\cup B,E)$, let $1OPT_d(G)$ denote the size of the largest \emph{one-sided} $d$-scattered set of $G$, i.e.\ a set that only includes vertices from the same side of the bipartition $A$ or $B$, but not both.
\end{definition}

\begin{lemma}\label{one_sided_opt}
 For a bipartite graph $G=(A\cup B,E)$, it is $1OPT_d(G)\ge OPT_d(G)/2$.
\end{lemma}
\begin{proof}
 Consider an optimal solution $S\subseteq A\cup B$ with $|S|=OPT_d(G)$. Then at least half of the vertices of $S$ are contained in one side of $G$, i.e.\ it is either $|S\cap A|\ge|S|/2$ or $|S\cap B|\ge|S|/2$ (or both if $|S\cap A|=|S\cap B|=|S|/2$). By definition, it is also $1OPT_d(G)\ge|S\cap A|$ and $1OPT_d(G)\ge|S\cap B|$. Thus in both cases it must be $1OPT_d(G)\ge OPT_d(G)/2$.
\end{proof}

\begin{theorem}\label{bipartite_even_thm}
 For any bipartite graph $G=(A\cup B,E)$ of size $n$ and $d$ even, the \dS\ problem can be approximated within a factor of $2\sqrt{n}$ in polynomial time.
\end{theorem}
\begin{proof}
 We will consider two cases based on the parity of $d/2$ and define appropriate \textsc{Set Packing} instances whose solutions are in a one-to-one correspondence with one-sided $d$-scattered sets in $G$. We will then be able to apply the $\sqrt{n}$-approximation for \textsc{Set Packing} of \cite{HalldorssonKT00}. We will repeat this process for both sides $A,B$ of the bipartition and retain the best solution found. Thus we will be able to approximate $1OPT_d(G)$ within a factor of $\sqrt{n}$ and then rely on Lemma~\ref{one_sided_opt} to obtain the claimed bound.
 
 Our \textsc{Set Packing} instances are defined as follows: for $d/2$ even, we make a set $c_i$ for every vertex $a_i$ of $A$ (i.e.\ from one side) and an element $e_j$ for every vertex $b_j$ of $B$ (i.e.\ from the other side). For $d/2$ odd, we make a set $c_i$ for every vertex $a_i$ of $A$ (again from one side) and an element $e_j$ for every vertex $b_j$ of $B$ \emph{and} an element $r_i$ for every vertex $a_i\in A$ (i.e.\ from both sides). Note that $i,j\le n$. In both cases we include an element corresponding to vertex $x\in G$ in a set corresponding to a vertex $y\in G$, if $d_G(x,y)\le d/2-1$. We then claim that for any given collection $C$ of compatible (i.e.\ non-overlapping) sets in the above definitions, we can always find a one-sided $d$-scattered set $S\subseteq A$ in $G$ with $|C|=|S|$ and vice-versa.
 
 First consider the case where $d/2$ is even. Given a one-sided $d$-scattered set $S\subseteq A$, we let $C$ include all the sets that correspond to some vertex in $S$ and suppose for a contradiction that there exists a pair of sets $c_1,c_2\in C$ that are incompatible, i.e.\ that there exists some element $e$ with $e\in c_1$ and $e\in c_2$. Let $a_1,a_2\in A$ be the vertices corresponding to sets $c_1,c_2$ and $b\in B$ be the vertex corresponding to element $e$. Then it must be $d_G(a_1,b)\le d/2-1$ since $e\in c_1$ and $d_G(b,a_2)\le d/2-1$ since $e\in c_2$, that gives $d_G(v_1,v_2)\le d-2$, which contradicts $S$ being a $d$-scattered set. On the other hand, given collection $C$ of compatible sets we let $S\subseteq A$ include all the vertices corresponding to some set in $C$ and suppose there exists a pair of vertices $a_1,a_2\in S$ for which it is $d_G(a_1,a_2)<d$. Since $d$ is even and $a_1,a_2\in A$, if $d_G(a_1,a_2)<d$ it must be $d_G(a_1,a_2)\le d-2$, as any shortest path between two vertices on the same side of a bipartite graph must be of even length. Thus there must exist at least one vertex $b\in B$ on a shortest path between $a_1,a_2$ in $G$ for which it is $d_G(a_1,b)\le d/2-1$ and $d_G(b,a_2)\le d/2-1$. This means that the element $e$ corresponding to vertex $b\in B$ must be included in both sets $c_1,c_2$ corresponding to vertices $a_1,a_2\in A$, which contradicts the compatibility of sets in $C$.
 
 We next consider the case where $d/2$ is odd. Given a one-sided $d$-scattered set $S\subseteq A$, we again let $C$ include all sets that correspond to some vertex in $S$. If there exists a pair of sets $c_1,c_s\in C$ that contain the same element $e$ corresponding to some vertex $b\in B$ or some element $r$ that corresponds to a vertex $a\in A$, then by the same argument as in the even case we know that there must exist paths of length $\le d/2-1$ from both vertices $a_1,a_2\in A$ (corresponding to $c_1,c_2\in C$) to vertex $b\in B$ or $a\in A$ and thus it must be $d_G(a_1,a_2)<d$. On the other hand, given a collection $C$ of compatible sets we again let $S\subseteq A$ include all the vertices corresponding to sets in $C$. Supposing there exists a pair $a_1,a_2\in S$ for which it is $d_G(a_1,a_2)<d$, then again as $d$ is even it must be $d_G(a_1,a_2)\le d-2$. This means there must be a vertex $a\in A$ on a shortest path between $a_1$ and $a_2$ for which $d_G(a_1,a)\le d/2-1$ and $d_G(a,a_2)\le d/2-1$, which means the corresponding sets $c_1,c_2\in C$ must both contain element $r$ that corresponds to this vertex $a\in A$, giving a contradiction.
 
 Our algorithm then is as follows. For a given bipartite graph $G=(A\cup B,E)$, we define an instance of \textsc{Set Packing} as described above (depending on the parity of $d/2$) and apply the $\sqrt{n}$-approximation of \cite{HalldorssonKT00}. Observe that $|A|,|B|\le n$. We then exchange the sets $A,B$ in the definitions of our instances and repeat the same process. This will return a solution $S$ of size $|S|\ge\frac{1OPT_d(G)}{\sqrt{n}}$, which by Lemma~\ref{one_sided_opt} is $\ge\frac{OPT_d(G)}{2\sqrt{n}}$. 
\end{proof}

\section{Super-polynomial time}\label{dSS_super-poly-sec}
This section concerns itself with running times that are not restricted to being functions polynomial in the size of the input. We begin with an upper bound on the size of the solution in any connected graph that is then employed in obtaining a simple exact exponential-time algorithm.


\begin{lemma}\label{dss_size_bound}
 The maximum size of any $d$-Scattered Set in a connected graph is $\left\lfloor\frac{n}{\lfloor d/2\rfloor}\right\rfloor$.
\end{lemma}
\begin{proof}
 Given connected graph $G=(V,E)$, let $S\subseteq V$ be a $d$-scattered set in $G$. To each $u\in S$, we will assign all vertices at distance $<\lfloor d/2\rfloor$: let $M(u)\coloneqq\{u\}\cup\{v\in V|d(u,v)<\lfloor d/2\rfloor\}$. Our aim is to show that for any $u\in S$, it must be $|M(u)|\ge\lfloor d/2\rfloor$. In other words, for any vertex $u$ in the solution there must be at least $\lfloor d/2\rfloor-1$ distinct vertices that are at distance $<\lfloor d/2\rfloor$ from $u$ and at distance $\ge\lfloor d/2\rfloor$ from any vertex $w\in S$. Observe that if for some pair $u,w\in S$, we have $M(u)\cap M(w)\not=\emptyset$, then $d(u,w)<d$, as there exists a vertex at distance $<\lfloor d/2\rfloor$ from both $u,w$.
 
 Consider some $u\in S$ and let $k$ be the number of vertex-disjoint paths of length $<\lfloor d/2\rfloor$ starting from $u$, i.e.\ such that no vertices are shared between them. Then it must be $|M(u)\setminus\{u\}|\ge k(\lfloor d/2\rfloor-1)$. If $V\subseteq M(u)$, then $OPT_d(G)=|S|=1$ and the claim trivially holds, so we may assume that there is at least one vertex $z\notin M(u)$. This means $k\ge1$, since $G$ is connected and there must be (at least) one path from $z$ to $u$ of length $\ge\lfloor d/2\rfloor$. It is then $|M(u)\setminus\{u\}|+1\ge \lfloor d/2\rfloor -1+1$ giving $|M(u)|\ge\lfloor d/2\rfloor$.
 
 This means that for each vertex $u$ taken in any solution $S$ there must be at least $\lfloor d/2\rfloor$ distinct vertices in the graph, i.e.\ $G$ must contain $|S|$ disjoint subsets of size at least $\lfloor d/2\rfloor$ and the claim follows.
\end{proof}

\begin{theorem}\label{dss_exact_thm}
 The \textsc{$d$-Scattered Set} problem can be solved in $O^*((ed)^{\frac{2n}{d}})$ time.
\end{theorem}
\begin{proof}
 We simply try all sets of vertices of size at most
 $\left\lfloor\frac{n}{\lfloor d/2\rfloor}\right\rfloor$ for feasibility and retain the best one found. By
 Lemma~\ref{dss_size_bound}, the optimal solution must be contained therein. The number of sets we examine is
 $\le\frac{2n}{d}{n \choose d/2}=O^*((ed)^{\frac{2n}{d}})$, that gives the running time.
\end{proof}

\subsection{Inapproximability}\label{dSS_super-poly-inapprox-sec}
We now turn our attention to the problem's hardness of approximation in super-polynomial time. We use Theorem~\ref{gap_reduction} in conjunction with straightforward reductions from \textsc{Independent Set} to \dS\ for the two cases, that depend on the parity of $d$.

\begin{theorem}\label{inapprox_even}
Under the randomized ETH, for any even $d\ge4$, $\epsilon>0$ and $\rho\le
(2n/d)^{5/6}$, no $\rho$-approximation for \textsc{$d$-Scattered Set} can take
time $2^{\left(\dfrac{n^{1-\epsilon}}{\rho^{1+\epsilon}d^{1-\epsilon}}\right)}\cdot n^{O(1)}$.
\end{theorem}
\begin{proof}
 We suppose the existence of a $\rho$-approximation algorithm for \dS\ of running time $2^{\left(\frac{n^{1-\epsilon}}{d^{1-\epsilon}\rho^{1+\epsilon}}\right)}\cdot n^{O(1)}$ for some $\epsilon>0$ and aim to show this would violate the (randomized) ETH. First let $\epsilon_1>0$ be such that $\epsilon>\epsilon_1$ and $\epsilon>\frac{2\epsilon_1}{1-3\epsilon_1}$. We next define $\epsilon_2=\frac{1}{1-\epsilon_1}-1$ and $r=\rho^{\frac{1-\epsilon_1}{1-3\epsilon_1}}$. Then, given a formula $\phi$ of \textsc{3SAT} on $N$ variables, we use the reduction of Theorem~\ref{gap_reduction} with parameters $r$ and $\epsilon_2$ to build a graph $G$ from $\phi$, with $|V(G)|=N^{1+\epsilon_2}r^{1+\epsilon_2}$ and maximum degree $r$, such that with high probability: if $\phi$ is satisfiable then $\alpha(G)\ge N^{1+\epsilon_2}r$; if $\phi$ is not satisfiable then $\alpha(G)\le N^{1+\epsilon_2}r^{2\epsilon_2}$. Thus an approximation algorithm with ratio $r^{1-2\epsilon_2}$ would permit us to decide if $\phi$ is satisfiable. We have that $r^{1-2\epsilon_2}=(\rho^{\frac{1-\epsilon_1}{1-3\epsilon_1}})^{1-2\epsilon_2}=(\rho^{\frac{1-\epsilon_1}{1-3\epsilon_1}})^{3-\frac{2}{1-\epsilon_1}}=\rho^{\frac{3+\frac{2\epsilon_1-2}{1-\epsilon_1}-3\epsilon_1}{1-3\epsilon_1}}=\rho$.

 We will construct graph $H$
 from $G$ as follows (similarly to Theorem~3.10 in \cite{HalldorssonKT00}, see also Figure \ref{fig:ds_reductions}):
 graph $H$ contains a copy of $G$ and a distinct path of $d/2-1$ edges
 attached to each vertex of $G$. Without loss of generality, we may assume that any $d$-scattered set will prefer selecting an endpoint of these attached paths than some vertex from $G$, as these selections would exclude strictly fewer vertices from the solution, i.e.\ any solution can only be improved by exchanging any vertex of $G$ with selecting the (other) endpoint of the path attached to it. A pair of vertices that are endpoints of such paths (and not originally in $G$) will be at distance $\ge2(d/2-1)+2=d$, only if the vertices of $G$ to which they are attached are non-adjacent, i.e.\ if the shortest path between them is of length at least 2.
 Thus, $d$-scattered sets  in $H$ are in one-to-one correspondence with independent sets in $G$ and
 $\alpha(G)=OPT_d(H)$. The size of $H$ is  $n=|V(H)|=|V(G)|(d/2-1+1)=N^{1+\epsilon_2}r^{1+\epsilon_2}(d/2)$.
 Note also that $\rho\le N^{5}$ and thus $\rho\le(2n/d)^{5/6}$.
 
 If $\phi$ is satisfiable then $OPT_d(H)=\alpha(G)\ge N^{1+\epsilon_2}r$, while if $\phi$ is not satisfiable then $OPT_d(H)=\alpha(G)\le N^{1+\epsilon_2}r^{2\epsilon_2}$. Therefore, applying the supposed $\rho$-approximation for \dS\ on $H$ would permit us to solve \textsc{3SAT} in time $2^{\left(\frac{n^{1-\epsilon}}{d^{1-\epsilon}\rho^{1+\epsilon}}\right)}\cdot n^{O(1)}$, with high probability. We next show that this would violate the ETH, i.e.\ $2^{\left(\frac{n^{1-\epsilon}}{d^{1-\epsilon}\rho^{1+\epsilon}}\right)}\cdot n^{O(1)}=2^{o(N)}$.
 We have:
 \begin{equation}
 n=N^{1+\epsilon_2}r^{1+\epsilon_2}(d/2)\Rightarrow
 N=\left(\frac{2n}{d}\right)^{1-\epsilon_1}\cdot\frac{1}{\rho^{\left(\frac{1-\epsilon_1}{1-3\epsilon_1}\right)}}
 \end{equation}
 And we then need to show:
 \begin{equation}
  2^{\left(\frac{n^{1-\epsilon}}{d^{1-\epsilon}\rho^{1+\epsilon}}\right)}\cdot n^{O(1)}=2^{o(N)}=2^{o\left(\left(\frac{2n}{d}\right)^{1-\epsilon_1}\cdot\frac{1}{\rho^{\left(\frac{1-\epsilon_1}{1-3\epsilon_1}\right)}}\right)}
 \end{equation}
 Observe that, since $\epsilon>\epsilon_1$ and $\epsilon>\frac{2\epsilon_1}{1-3\epsilon_1}$, it is:
 \begin{align}
 (n/d)^{1-\epsilon}<(2n/d)^{1-\epsilon_1}\\
 (1/\rho^{1+\epsilon})<\left(1/\rho^{(\frac{1-\epsilon_1}{1-3\epsilon_1})}\right)
 \end{align}
 Which then gives:
 \begin{equation}
  \lim_{(n,\rho)\rightarrow\infty}\frac{2^{\left((\frac{n}{d})^{1-\epsilon}\cdot\frac{1}{\rho^{1+\epsilon}}\right)}}{2^{\left((\frac{2n}{d})^{1-\epsilon_1}\cdot\frac{1}{\rho^{\left(\frac{1-\epsilon_1}{1-3\epsilon_1}\right)}}\right)}}=0
 \end{equation}
\end{proof}

The following reduction from \textsc{Independent Set} to \dS\ for odd values of $d$ uses a construction that includes a copy of every edge of the original graph (an \emph{edge gadget}, see Figure \ref{fig:ds_reductions}). This necessity is responsible for the difference in running times and is due to the parity idiosyncrasies of the problem as discussed above.

\begin{figure}[htbp]
 \centerline{\includegraphics[width=120mm]{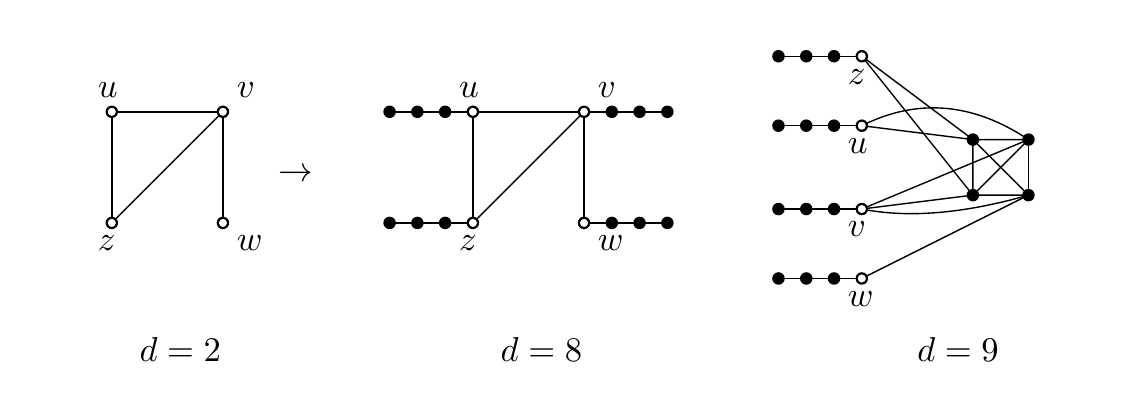}}
 \caption{Examples of the constructions for even (center) and odd (right) values of $d$. Note the existence of an edge ``gadget'' for the odd case.}
 \label{fig:ds_reductions}
 \end{figure}

\begin{theorem}\label{inapprox_odd}
Under the randomized ETH, for any odd $d\ge5$, $\epsilon>0$ and $\rho\le(2n/d)^{5/6}$, no $\rho$-approximation for \textsc{$d$-Scattered Set} can take
time $2^{\left(\dfrac{n^{1-\epsilon}}{\rho^{1+\epsilon}(d+\rho)^{1+\epsilon}}\right)}\cdot n^{O(1)}$.
\end{theorem}
\begin{proof}
 We suppose the existence of a $\rho$-approximation algorithm for \dS\ of running time $2^{\left(\dfrac{n^{1-\epsilon}}{\rho^{1+\epsilon}(d+\rho)^{1+\epsilon}}\right)}\cdot n^{O(1)}$ for some $\epsilon>0$ and aim to show this would violate the (randomized) ETH. We let $\epsilon_1>0$ be such that $\epsilon>\frac{\epsilon_1}{1+\epsilon_1}$ 
 and also $\epsilon>\frac{2\epsilon_1}{1-2\epsilon_1}$, as well as $\epsilon>\frac{2\epsilon_1^2+\epsilon_1}{1-2\epsilon_1^2-\epsilon_1}$.
 We then set $r=\rho^{\left(\frac{1}{1-2\epsilon_1}\right)}$. Given a formula $\phi$ of 3SAT on $N$ variables, we build graph $G$ from $\phi$ using the reduction of Theorem~\ref{gap_reduction} with parameters $r$ and $\epsilon_1$. The size of $G$ is $|V(G)|=N^{1+\epsilon_1}r^{1+\epsilon_1}$, its maximum degree is $r$ and with high probability: if $\phi$ is satisfiable then $\alpha(G)\ge N^{1+\epsilon_1}r$; if $\phi$ is not satisfiable then $\alpha(G)\le N^{1+\epsilon_1}r^{2\epsilon_1}$. Therefore an approximation algorithm with ratio $r^{1-2\epsilon_1}=\rho$ would permit us to decide if $\phi$ is satisfiable.

 For odd $d\ge5$, we will construct graph $H$ from $G$ as follows (again, a similar reduction is alluded to in the proof of Theorem~3.10 from \cite{HalldorssonKT00} and partly also used for Corollary~1 from \cite{EtoGM14}): we make a vertex in $H$ for each vertex of $G$ and we also attach a distinct path of $(d-3)/2$ edges to each of them. We then make a vertex for every edge of $G$, turn all these vertices into a clique and also connect each one to the two vertices of $H$ representing its endpoints. In this way, all pairs of vertices in $H$ are at distance $\le2(d-1)/2+1=d$ and the only vertices at exactly this distance are pairs of leaves on paths added to vertices that do not share a common neighbor representing some edge of $G$. Thus, $d$-scattered sets in $H$ are again in one-to-one correspondence with independent sets in $G$ and $\alpha(G)=OPT_d(H)$. 
 The size of $H$ is $n=|V(H)|=|V(G)|(d-1)/2+|E(G)|$. The construction of \cite{Chalermsook13} builds a graph where every vertex has degree at least one and at most $r$, therefore $|E(G)|\ge |V(G)|$ and $|E(G)|\le|V(G)|r/2$, that gives $n\le N^{1+\epsilon_1}r^{1+\epsilon_1}(\frac{d+r-1}{2})$, while $\rho\le N^{5}$, with $n\ge N\rho(d+1)/2$, that gives $\rho\le(2n/d)^{5/6}$.
 
 If $\phi$ is satisfiable then $OPT_d(H)=\alpha(G)\ge N^{1+\epsilon_1}r$, while if $\phi$ is not satisfiable then $OPT_d(H)=\alpha(G)\le N^{1+\epsilon_1}r^{2\epsilon_1}$. Thus the supposed $\rho$-approximation for \dS\ on $H$ would permit us to solve 3SAT in time $2^{\left(\dfrac{n^{1-\epsilon}}{\rho^{1+\epsilon}(d+\rho)^{1+\epsilon}}\right)}\cdot n^{O(1)}$, with high probability. We next show that this would violate the ETH, i.e.\ $2^{\left(\dfrac{n^{1-\epsilon}}{\rho^{1+\epsilon}(d+\rho)^{1+\epsilon}}\right)}\cdot n^{O(1)}=2^{o(N)}$.
 It is:
 \begin{align}
  n\le N^{1+\epsilon_1}r^{1+\epsilon_1}\left(\frac{d+r-1}{2}\right)&\Rightarrow\\
  \Rightarrow 2^N\ge2^{\left(\left(\frac{2n}{d+r-1}\right)^{\frac{1}{1+\epsilon_1}}\cdot\frac{1}{r}\right)}&=2^{\left(\left(\frac{2n}{d+\rho^{\left(\frac{1}{1-2\epsilon_1}\right)}-1}\right)^{\frac{1}{1+\epsilon_1}}\cdot\frac{1}{\rho^{\left(\frac{1}{1-2\epsilon_1}\right)}}\right)}
 \end{align}
  Observe it is $(d+\rho)^{\frac{1}{1-2\epsilon_1}}>(d+\rho^{\frac{1}{1-2\epsilon_1}}-1)$ and so $2^N>2^{\left(\left(\frac{2n}{(d+\rho)^{\frac{1}{1-2\epsilon_1}}}\right)^{\frac{1}{1+\epsilon_1}}\cdot\frac{1}{\rho^{\left(\frac{1}{1-2\epsilon_1}\right)}}\right)}$.
 We thus then require:
 \begin{equation}
  \lim_{(n,\rho)\rightarrow\infty}\frac{2^{\left(\dfrac{n^{1-\epsilon}}{\rho^{1+\epsilon}(d+\rho)^{1+\epsilon}}\right)}}{2^{\left({\left(\frac{2n}{(d+\rho)^{\frac{1}{1-2\epsilon_1}}}\right)^{\frac{1}{1+\epsilon_1}}\cdot\frac{1}{\rho^{\left(\frac{1}{1-2\epsilon_1}\right)}}}\right)}}=0
 \end{equation}
This is shown by the following inequalities:
\begin{align}
 \epsilon>\frac{\epsilon_1}{1+\epsilon_1}&\Rightarrow n^{1-\epsilon}<2n^{\frac{1}{1+\epsilon_1}}\\
 \epsilon>\frac{2\epsilon_1^2+\epsilon_1}{1-2\epsilon_1^2-\epsilon_1}&\Rightarrow\frac{1}{(d+\rho)^{1+\epsilon}}<\frac{1}{(d+\rho)^{\frac{1}{(1+\epsilon_1)(1-2\epsilon_1)}}}\\
 \epsilon>\frac{2\epsilon_1}{1-2\epsilon_1}&\Rightarrow\frac{1}{\rho^{1+\epsilon}}<\frac{1}{\rho^{(\frac{1}{1-2\epsilon_1})}}
 \end{align}
%
%
%
\end{proof}

\subsection{Approximation}\label{dSS_super-poly-approx-sec}
We complement the above hardness results with approximation algorithms of almost matching super-polynomial running times. Similarly to the exact algorithm of Theorem~\ref{dss_exact_thm}, the upper bound from the beginning of this section is used for even values of $d$, while for the odd values this idea is combined with a greedy scheme based on minimum vertex degree.

\begin{theorem}\label{approx_even}
 For any even $d\ge2$ and any $\rho\le\frac{n}{\lfloor d/2\rfloor}$, there is a $\rho$-approximation algorithm for \dS\ of running time $O^*((e\rho d)^{\frac{2n}{\rho d}})$.
\end{theorem}
\begin{proof}
 From Lemma~\ref{dss_size_bound} we know that the maximum size of a $d$-scattered set is $\lfloor\frac{n}{\lfloor d/2\rfloor}\rfloor$. We thus simply try all sets of vertices of size at most $\frac{n}{\rho\lfloor d/2\rfloor}$ for feasibility and retain the best one: these are $\le\frac{n}{\rho\lfloor d/2\rfloor}{n \choose \rho\lfloor d/2\rfloor}=O^*((e\rho d)^{\frac{2n}{\rho d}})$, that gives the running time. If the graph is not connected, we can apply Lemma~\ref{dss_size_bound} to each connected component $C$ of size $n_C$ and then consider all subsets of size at most $\frac{n_C}{\rho\lfloor d/2\rfloor}$ in each $C$.
\end{proof}

\begin{theorem}\label{approx_odd}
 For any odd $d\ge3$ and any $\rho\le\frac{n}{\lfloor d/2\rfloor}$, there is a $\rho$-approximation algorithm for \dS\ of running time $O^*((e\rho d)^{\frac{2n}{\rho(d+\rho)}})$.
\end{theorem}
\begin{proof}
 Let $q=(d-1)/2$ and $G'$ be the $q$-th power of graph $G$. We then claim that any $d$-scattered set $S$ in $G$ is a $3$-scattered set in $G'$ and vice-versa: if $S$ is a $d$-scattered set in $G$, then for any pair $u,v\in S$, it is $d_G(u,v)\ge d$. For some pair of distinct vertices $w,z$ on a shortest path between $u,v$ in $G$, it must be $d_G(u,w)=d_G(z,v)=q$, while also $d_G(u,z)=d_G(w,v)>q$, i.e.\ $w$ and $z$ are two vertices on a shortest path from $u$ to $v$, each at equal distance $q$ from their closest endpoint ($v$ or $u$).
 Then $d_{G'}(u,w)=d_{G'}(z,v)=1$. Since $w,z$ are distinct, it must be $d_{G'}(w,z)\ge1$ which gives also $d_{G'}(u,v)\ge3$, since $d_{G'}(u,z)=d_{G'}(w,v)>1$.
 
 If $S$ is a $3$-scattered set in $G'$, then for any pair $u,v\in S$ it is $d_{G'}(u,v)\ge3$. Now, any shortest path in $G$ between $u,v$ must contain two distinct vertices $w,z$ for which $d_G(u,w)=q$ and $d_G(z,v)=q$, while $d_G(u,z)=d_G(w,v)>q$. If no such pair of vertices exists in $G$, then $d_{G'}(u,v)<3$: any pair of vertices at distance $\le q$ in $G$ are adjacent in $G'$ and so for the distance between $u,v$ in $G'$ to be at least 3, there must be two vertices each at distance $\ge q$ from $u,v$ in $G$. From this we get that $d_G(u,v)\ge 2q+1=d$, since $d_G(w,z)\ge1$, as $w\not= z$.
 
 The algorithm then proceeds in two phases. For the first phase, so long as there exists an unmarked vertex $v_i$ in $G'$ of minimum degree $<\rho$, we mark $v_i$ as `selected' and add it to $S_1\subseteq V$, marking all vertices at distance $\le2$ from $v_i$ in $G'$ as `excluded' and adding them to $X_i\subseteq V$. That is, $X_i=N^2_{G'}(v_i)$ and we let $X=N^2_{G'}(S_1)$, i.e.\ $X=X_1\cup\dots\cup X_{|S_1|}$. We also let the remaining (unmarked) vertices belong to $H\subseteq V$. Thus when this procedure terminates we have $V$ partitioned into three sets $S_1,X,H$, while the degree of any vertex in $H$ is $\ge\rho$. For the second phase, we try all subsets of vertices of $H$ of size at most $\frac{2n}{\rho(\rho+\lfloor d/2\rfloor)}$ for feasibility and retain the best one. These are $\le\frac{2n}{\rho(\rho+\lfloor d/2\rfloor)}{n\choose \rho/2(\rho+\lfloor d/2\rfloor)}=O^*((e\rho d)^{2n/\rho(d+\rho)})$, giving the upper bound on the running time.
 
 Now let $S^*_1$ be a 3-scattered set of maximum size in the subgraph of $G'$ induced by $S_1\cup X$, i.e.\ $|S^*_1|=OPT_3(G'[S_1\cup X])$ and $S^*_2$ be a 3-scattered set of maximum size in the subgraph of $G'$ induced by $H$, i.e.\ $|S^*_2|=OPT_3(G'[H])$. As the degree of any vertex $v_i\in S_1$ is $<\rho$, we have (1): $|S^*_1|<\rho|S_1|$, since for every vertex $u$ in $N^1_{G'}(v_i)$, for $v_i\in S_1$, 3-scattered set $S^*_1$ can contain at most one vertex $w$ from $N^1_{G'}(u)$, as the distance between $w$ and another neighbor of $u$ is $\le2$. For the second phase, it must be $|S^*_2|\le n/\rho\Rightarrow1/|S^*_2|\ge\rho/n$, since all vertices of $H$ are of degree $\ge\rho$ and these neighborhoods are disjoint: if two vertices of $S^*_2$ share a common neighbor then they cannot belong in a 3-scattered set. From Lemma~\ref{dss_size_bound}, we also know that $|S^*_2|\le n/\lfloor d/2\rfloor\Rightarrow1/|S^*_2|\ge\lfloor d/2\rfloor/n$. Adding the two inequalities gives $2/|S^*_2|\ge\frac{\rho+\lfloor d/2\rfloor}{n}\Rightarrow|S^*_2|\le\frac{2n}{\rho+\lfloor d/2\rfloor}$. Furthermore, it is $|S_2|\le\frac{2n}{\rho(\rho+\lfloor d/2\rfloor)}$, by construction. Dividing the two inequalities gives (2): $\frac{|S^*_2|}{|S_2|}\le\rho\Rightarrow |S^*_2|\le\rho|S_2|$.
 From (1) and (2) we get that $|S^*_1|+|S^*_2|\le\rho(|S_1|+|S_2|)$. It is $OPT_d(G)=OPT_3(G')\le|S^*_1|+|S^*_2|$ since $S_1\cup X$ and $H$ form a partition of $G'$. Our algorithm returns a solution of size $|S_1|+|S_2|$ and thus our approximation ratio is $\frac{OPT_d(G)}{|S_1|+|S_2|}\le\frac{|S^*_1|+|S^*_2|}{|S_1|+|S_2|}\le\rho$. If the graph is not connected, we can apply Lemma~\ref{dss_size_bound} to each connected component $C$ of size $n_C$ and then try all subsets of size at most $\frac{n_C}{\rho\lfloor d/2\rfloor}$ in each $C$ and obtain an additive version of (2) for each component. 
\end{proof}

\subsection*{Treewidth of power graphs}
We close this paper with a note on the treewidth of power graphs. Similar ideas as those used in the above results also point to the following upper bound on the increase in treewidth taking place when computing the power of a graph of bounded degree:

\begin{theorem}\label{tw_UB_thm}
 For any graph $G$ of treewidth $tw$ and maximum degree bounded by $\Delta$, the treewidth $\tw'$ of the $d$-th power $G^d$ is at most $\tw'\le \tw\cdot\Delta\sum_{i=0}^{d/2-1}(\Delta-1)^i=O(\tw\cdot\Delta^{d/2})$.
\end{theorem}
\begin{proof}
 Given a tree decomposition $T$ of $G=(V,E)$ of width $\tw$, we make a tree decomposition $T'$ of $G^d=(V,E^d)$ by replacing the appearance of each vertex $v$ in each bag of $T$ with $v$ and the set of vertices at distance at most $d/2$ from $v$ in $G$, i.e.\ with $N^{d/2}_G(v)\cup\{v\}$. It is $|N^{d/2}_G(v)|\le\sum_{i=0}^{d/2-1}(\Delta(\Delta-1)^i)$, from which we get the upper bound. This is a valid tree decomposition for $G^d$ as: (a) all vertices appear in some bag of $T'$ as they appeared in $T$, (b) for every edge $(u,v)$ in $G^d$, either $(u,v)\in E$ and there is a bag in $T$ containing both $u,v$ and thus there is one also in $T'$, or $(u,v)$ was added to $E^d$ due to the distance between $u,v$ being $\le d$ in $G$. In this case there must be at least one vertex $w$ at distance $\le d/2$ from both $u$ and $v$ in $G$, meaning there will be a bag in $T'$ containing all three vertices $u,v,w$ that was constructed from a bag of $T$ that contains $w$.
 
 Finally, (c) for every vertex $v$ appearing in two bags $X',Y'$ of $T'$, vertex $v$ also appears on every bag on the path from $X'$ to $Y'$ in $T'$: consider (for a contradiction) the existence of a bag $Z'$ on the path from $X'$ to $Y'$ in $T'$ that does not contain $v$,  and let $X,Y,Z$ be the corresponding bags in $T$. Since $X',Y'$ contain $v$, then both $X$ and $Y$ contain some vertex $u$ at distance at most $d/2$ from $v$ in $G$, or $v$ itself, i.e.\ $u\in N_G^{d/2}(v)\cup\{v\}$. If both $X$ and $Y$ contain $v$, then as $T$ is a valid tree decomposition, so does $Z$ and therefore also $Z'$. Thus we may assume that at least one of $X,Y$ do not contain $v$, as well as $v\notin Z$.
 As $Z$ is a separator, then $v$ must appear only on one side of $Z$ in $T$. We assume (without loss of generality) that $v$ only appears on the $X$-side of $T$ (from $Z$) and $v$ is not contained in $Y$. Thus $Y$ must contain some vertex $u$ at distance $\le d/2$ from $v$ in $G$. As $Z$ is a separator, the path from $v$ to $u$ must contain at least one vertex $w\in Z$, at distance $<d/2$ from $v$. Thus $Z'$ must also contain $v$, as it includes all vertices at distance $\le d/2$ from $w$.
\end{proof}

As for graphs of maximum degree bounded by $\Delta$, we have $\cw\le O(\Delta\cdot\tw)$ (see \cite{COURCELLE2012866}) and $\tw\le O(\Delta\cdot\cw)$ (directly derived from the well-known result of Gurski and Wanke \cite{GurskiW00}), we also obtain the following corollary.

\begin{corollary}\label{cw_UB_cor}
 For any graph $G$ of clique-width $\cw$ and maximum degree bounded by $\Delta$, the clique-width $\cw'$ of the $d$-th power $G^d$ is at most $\cw'\le O(\cw\cdot\Delta^{d/2+2})$.
\end{corollary}

\section{Conclusion}\label{sec_conc}
In this paper we furthered our understanding of the \dS\ problem by answering some remaining questions on its (super-)polynomial (in-)approximability. In particular, we showed the following:
\begin{itemize}
 \item A lower bound of $\Delta^{\lfloor d/2\rfloor-\epsilon}$ on the approximation ratio of any polynomial-time algorithm for graphs of maximum degree $\Delta$ and an improved upper bound of $O(\Delta^{\lfloor d/2\rfloor})$ on the approximation ratio of any greedy scheme for this problem, that matches our lower bound.
 \item A polynomial-time approximation algorithm of ratio $2\sqrt{n}$ for bipartite graphs and even values of $d$, that complements known results by considering the only remaining open case.
 \item An exact exponential-time algorithm of complexity $O^*((ed)^{\frac{2n}{d}})$, based on an upper bound on the size of any solution.
 \item A lower bound on the complexity of any $\rho$-approximation algorithm of (roughly) $2^{\frac{n^{1-\epsilon}}{\rho d}}$ for even $d$ and $2^{\frac{n^{1-\epsilon}}{\rho(d+\rho)}}$ for odd $d$, under the randomized ETH.
 \item $\rho$-approximation algorithms of running times $O^*((e\rho d)^{\frac{2n}{\rho d}})$ for even $d$ and $O^*((e\rho d)^{\frac{2n}{\rho(d+\rho)}})$ for odd $d$ that (almost) match the above lower bounds, thus giving a clear picture of the trade-off curve between approximation and running time.
\end{itemize}
Apart from the possibility of ``de-randomization'' of the results above that use the randomized construction of \cite{Chalermsook13} as a starting point, some remaining unanswered questions here would concern the complexity of the problem on chordal bipartite graphs (also mentioned as an open problem by \cite{EtoGM14}) as well as the functionality of the PTAS for planar graphs by the same authors, that only works for fixed values of $d$ as it extends the well-known approach of \cite{Baker94} for obtaining such algorithms in planar graphs for e.g.\ \textsc{Independent Set}. Because this approach involves breaking down the graph into (roughly) $d$-outerplanar subgraphs and then exactly solving the problem in each of these using dynamic programming over their tree decompositions, for values of $d$ that are not constant (say $d\ge\sqrt{n})$) this is not achievable in polynomial-time, due to the exponent of the treewidth algorithms depending on $d$. It would be interesting to see an extension of this (or some other) approach for the case of unbounded $d$, or, conversely, a hardness reduction proving it is unlikely. The difficult part here would have to involve a construction that is very efficient in terms of \emph{crossing} gadgets in order to maintain planarity, or, from the other side, a way to optimally solve the problem in carefully constructed subgraphs without requiring exponential time on $d$.

\bibliography{bibliography}

\end{document}